\documentclass{article}

\usepackage{graphicx,url}

\title{Non-commutative propositional logic with short-circuited biconditional and NAND}

\author{
	Dalia Papuc$^1$ \& Alban Ponse$^2$\\[2mm]
  {\small 1. School of Computer and Communication Sciences, EPFL}\\
  {\small 2. Section Theory of Computer Science, Informatics Institute}\\[0mm]
  {\small  Faculty of Science, University of Amsterdam}\\[0mm]
  {\small \url{https://daliapapuc.com/} \& \url{https://staff.science.uva.nl/a.ponse/}
}
}

\date{}

\usepackage{hyperref}
\usepackage{setspace,tikz,stmaryrd,amssymb,amsmath,amsthm,url,color,marginnote}
\usepackage{tocloft}

\newcommand{\tr}{\ensuremath{{\sf T}}}
\newcommand{\fa}{\ensuremath{{\sf F}}}
\newcommand{\true}{\ensuremath{\textit{true}}}
\newcommand{\false}{\ensuremath{\textit{false}}}

\newcommand{\lef}{\ensuremath{\triangleleft}}
\newcommand{\rig}{\ensuremath{\triangleright}}
\newcommand{\SigCP}{\ensuremath{\Sigma_{\textup{CP}}(A)}}
\newcommand{\SigSCL}{\ensuremath{\Sigma_{\textup{SCL}}(A)}}
\newcommand{\SigIff}{\ensuremath{\Sigma_{\textup{SCL$\ell$I}}(A)}}
\newcommand{\SigIffD}{\ensuremath{\Sigma_{\textup{SCL$\ell$I$\ell$X}}(A)}}
\newcommand{\SigSCLu}{\ensuremath{\Sigma_{\textup{SCL}^\und}(A)}}
\newcommand{\SigMix}{\ensuremath{\Sigma{(A)}}}

\newcommand{\axname}[1]{\ensuremath{\textup{\textrm{#1}}}}
\newcommand{\CP}{\axname{CP}}
\newcommand{\CPmem}{\axname{\CP$_{\mem}$}}
\newcommand{\CPmemu}{\ensuremath{\axname{\CP}^\und_{\mem}}}
\newcommand{\MSCLeu}{\ensuremath{\axname{EqMSCL}^\und}}
\newcommand{\MSCLu}{\ensuremath{\axname{MSCL}^\und}}
\newcommand{\und}{\ensuremath{{\mathsf U}}}
\newcommand{\MSCLe}{\axname{EqMSCL}}
\newcommand{\MSCLea}{\axname{EqMSCL$_{\ell\text{I}}$}}
\newcommand{\MSCLeau}{\axname{EqMSCL$^\und_{\ell\text{I}}$}}
\newcommand{\MSCLead}{\axname{EqMSCL$_{\ell\text{I}\ell\text{X}}$}}
\newcommand{\MSCL}{\axname{MSCL}}
\newcommand{\SCLe}{\axname{EqFSCL}}
\newcommand{\NANDe}{\axname{EqMSCL$_{\ell\text{N}}$}}
\newcommand{\NANDeu}{\axname{EqMSCL$^\und_{\ell\text{N}}$}}
\newcommand{\FSCL}{\axname{FSCL}}

\newcommand{\Terms}{\ensuremath{{\textit{T}_A}}}

\newcommand{\MSNFu}{\ensuremath{\text{mUNBF}\,}}
\newcommand{\MFu}{\ensuremath{\textit{MUNBF}_A\,}}
\newcommand{\MSNFus}{\ensuremath{\text{mUNBFs}\,}}
\newcommand{\BF}{\ensuremath{\textit{BF}_A}}
\newcommand{\MBF}{\ensuremath{\textit{MBF}_A}}
\newcommand{\MBFu}{\ensuremath{\textit{MBF}_A^{\;\und}}}
\newcommand{\baf}{\ensuremath{\mathit{bf}}}
\newcommand{\mem}{\ensuremath{{mem}}}
\newcommand{\membf}{\ensuremath{\mathit{mbf}}}
\newcommand{\memf}{\ensuremath{\mathit{mf}}}
\newcommand{\PS}{\ensuremath{\textit{C}_A}}
\newcommand{\PSu}{\ensuremath{\textit{C}_A^{\;\und}}}

\newcommand{\ri}{r}
\newcommand{\dl}{\textit{dl}}
   
\newcommand{\leftand}{~
     \mathbin{\setlength{\unitlength}{1ex}
     \begin{picture}(1.4,1.8)(-.3,0)
     \put(-.6,0){$\wedge$}
     \put(-.54,-0.2){\textcolor{white}{\circle*{0.6}}}
     \put(-.54,-0.2){\circle{0.6}}
     \end{picture}
     }}
 
\newcommand{\fulland}{~
     \mathbin{\setlength{\unitlength}{1ex}
     \begin{picture}(1.4,1.8)(-.3,0)
     \put(-.6,0){$\wedge$}
     \put(-.54,-0.2){\circle*{0.6}}
     \end{picture}
     }}

\newcommand{\leftor}{~
     \mathbin{\setlength{\unitlength}{1ex}
     \begin{picture}(1.4,1.8)(-.3,0)
     \put(-.6,0){$\vee$}
     \put(-.54,1.54){\textcolor{white}{\circle*{0.6}}}
     \put(-.54,1.54){\circle{0.6}}
     \end{picture}
     }}
          
\newcommand{\lxor}{~
     \mathbin{\setlength{\unitlength}{1ex}
     \begin{picture}(1.5,1.8)(-.5,0)
     \put(-.8,0){$\oplus$}
     \put(-1.02,0.56){\circle{0.6}}
     \end{picture}
     }}
     
\newcommand{\liff}{~
     \mathbin{\setlength{\unitlength}{1ex}
     \begin{picture}(1.9,1.8)(-.5,0)\linethickness{0.16mm}
     \put(-.8,0){$\leftrightarrow$}
     \put(-.98,0.58){\circle{0.6}}
          \end{picture}
     }}

\newcommand{\lnand}{~
     \mathbin{\setlength{\unitlength}{1ex}
     \begin{picture}(1.1,1.8)(-.9,0)
     \put(-.8,0){$\mid$}
     \put(-.82,.56){\circle{0.6}}
     \end{picture}
     }}

\newcommand{\lnor}{~
     \mathbin{\setlength{\unitlength}{1ex}
     \begin{picture}(1.1,1.8)(-.9,0)
     \put(-.8,-.20){$\downarrow$}
     \put(-.62,.76){\circle{0.6}}
     \end{picture}
     }}

\newtheorem{theorem}{Theorem}[section]
\newtheorem{lemma}[theorem]{Lemma}  
\newtheorem{definition}[theorem]{Definition}  
\theoremstyle{definition}
  
\newtheorem{example}[theorem]{Example}

\begin{document}
\maketitle
\thispagestyle{empty}

\begin{abstract}
Short-circuit evaluation denotes the semantics of 
propositional connectives in which the second
argument is evaluated only if the first argument does not suffice 
to determine the value of the expression.
In programming, short-circuit evaluation is widely used, with
left-sequential conjunction and disjunction as primitive connectives.

We consider left-sequential, non-commutative propositional logic, also known as MSCL (memorising 
short-circuit logic), and start from a previously published, equational axiomatisation. 
First, we extend this logic with a left-sequential version of the biconditional connective, which 
allows for an elegant axiomatisation of MSCL. 
Next, we consider a left-sequential version of the NAND operator (the Sheffer stroke) and 
again give a complete, equational axiomatisation of the corresponding variant of MSCL. 
Finally, we consider these logical systems in a three-valued setting with a constant for 
`undefined', and again provide completeness results.
\\[2mm]
\emph{Keywords:}
Non-commutative propositional logic,
conditional connective,
sequential connective,
NAND,
short-circuit evaluation,
proposition algebra
\end{abstract}

\setlength{\cftbeforesecskip}{7pt}
{\small \tableofcontents}

\section{Introduction}
This paper is about non-commutative propositional logic, also known as Memorising Short-Circuit Logic (\MSCL),
enriched with some alternative connectives. 
In~\cite{BPS13}, \MSCL\ is defined 
as a logic for equational reasoning about sequential propositions with the property that atomic side 
effects do not occur, so that in the evaluation of a compound statement, the first evaluation result of each atom is 
memorised. Furthermore, the prescribed evaluation strategy
is \emph{short-circuit evaluation}: the second argument in
a conjunction or a disjunction is evaluated only if the first argument does not suffice 
to determine its evaluation result. In \MSCL, the binary connectives are left-sequential and written as
\[\leftand\quad\text{and}\quad\leftor, \]
where the little circle  prescribes that the left-argument must be evaluated first. 
We define these connectives using Hoare's \emph{conditional}, a ternary connective that naturally
prescribes short-circuit evaluation (in an if-then-else manner), and which has an elegant duality property. 

Left-sequential conjunction ${\leftand}$ and disjunction $\leftor$ are not commutative, 
but many common equational laws hold, such as the double negation shift, and idempotence and left-distributivity
of the binary connectives. 

We first extend \MSCL\ with the connective $\liff$, a left-sequential version of the biconditional 
connective $\leftrightarrow$, which in the MSCL-setting can be defined in different, equivalent ways:
\[
x\liff y=(\neg x\leftor y)\leftand(x\leftor \neg y)=(x\leftand y)\leftor(\neg x\leftand \neg y).
\]
We provide a complete, equational axiomatisation for this extension.
The combination of $\liff$ with ${\leftand}$, ${\leftor}$ and negation allows for an equational axiomatisation 
that is simple and perhaps more natural than the one for \MSCL.

Then, we consider an alternative for \MSCL, based on a left-sequential version of the NAND-operator 
(the Sheffer stroke) and again provide a complete, equational axiomatisation.

Finally, following~\cite{BPS21}, we consider extensions of these logical systems with \und, 
a constant for the truth value ``undefined". We discuss an advantage of the resulting NAND-system.

\textbf{Structure of the paper.}
In Section \ref{sec:2}, Hoare's ternary \emph{conditional} and `basic forms'
for this connective are discussed, and the left-sequential connectives are defined with this connective.
\\-
In Section \ref{sec:MVC} memorising valuation congruence is discussed, the congruence that characterises
\MSCL, and a normalisation function for so-called `mem-basic forms'.
\\-
In Section \ref{sec:mscl}, we discuss in detail the equational axiomatisation for \MSCL\ that is the 
starting point for this paper.
\\-
In Section \ref{sec:bico} we introduce a left-sequential variant of the biconditional connective
and give an equational axiomatisation for the extension of \MSCL\ with this connective. 
We prove a correspondence result, and briefly discuss duality.
\\-
In Section \ref{sec:nand} we discuss a left-sequential variant of the NAND connective, provide axioms for 
the memorising variant of this
extension, and prove a correspondence result.
\\-
In Section \ref{sec:three} we consider the extension of the new logical systems with \und.
\\-
In Section \ref{sec:conc} we present some conclusions and discuss related work.

All derivations from equational axiomatisations were found by the theorem prover 
\emph{Prover9}, and all independence results were found by the tool \emph{Mace4}, 
see~\cite{Prover9} for both tools.

\section{The conditional, basic forms, and propositional connectives}
\label{sec:2}
In this section we recall Hoare's ternary \emph{conditional} connective and `basic forms'
for the conditional. Next, we define sequential versions of the common propositional connectives using
this connective.

\smallskip

In 1985, Hoare introduced in~\cite{Hoa85} the ternary \emph{conditional} connective 
$p\lef q\rig r$ in order to express ``\texttt{if $q$ then $p$ else $r$}'' \footnote{%
  However, in 1948, Church introduced in~\cite{Chu48} the \emph{conditioned disjunction}
  connective $[p,q,r]$, which, following  the author, may be read ``$p$ or $r$ according 
  as $q$ or not $q$.'' and which expresses exactly the same connective as Hoare's conditional. 
  We further discuss this in Section~\ref{sec:conc} (Related work).}
and provided eleven equational axioms to show that the conditional and the two constants \tr\
and \fa\ for 
truth and falsehood characterise the propositional calculus.\footnote{%
  In 1963, Dicker provided in~\cite{Dic63} a set of five independent and elegant axioms 
  for the conditioned disjunction. Unaware of this,
  we provided in~\cite{BPS21} a set of four simple, independent axioms that is also complete.}

We are interested in both duality and equational axioms, and 
moreover in the short-circuit evaluation strategy suggested by the if-then-else reading that
prescribes that in $p\lef q\rig r$, 
\emph{first} $q$ is evaluated,
and only \emph{then} either $p$ or $r$ to determine the result of the evaluation.

Throughout this paper let $A$ be a set of atoms (propositional variables). The signature we consider is 
$\SigCP=\{\_\lef \_ \rig \_,\tr,\fa, a\mid a\in A\}$ and we write
\[\PS\]
for the set of closed terms over this signature.
Table~\ref{tab:CP} displays of a set of equational axioms for terms over
this signature, and we will refer to these axioms as \CP\ (for Conditional
Propositions).

\begin{table}
\hrule
\begin{align*}
\label{CP1}
\tag{CP1} x \lef \tr \rig y &= x\\
\label{CP2}\tag{CP2}
x \lef \fa \rig y &= y\\
\label{CP3}\tag{CP3}
\tr \lef x \rig \fa  &= x\\
\label{CP4}\tag{CP4}
\qquad
    x \lef (y \lef z \rig u)\rig v &= 
	(x \lef y \rig v) \lef z \rig (x \lef u \rig v)
\end{align*}
\hrule
\caption{\CP\ (Conditional Propositions), a set of equational axioms for 
free valuation congruence}
\label{tab:CP}
\end{table}

The dual of a closed term $P\in\PS$, notation $P^{\dl}$, is defined as follows 
(for $a\in A$):
\begin{align*}
\tr^{\dl}&=\fa,
&a^{\dl}&=a,\\
\fa^{\dl}&=\tr,
&(P\lef Q\rig R)^{\dl}&= R^{\dl}\lef Q^\dl\rig P^{\dl}.
\end{align*} 
The duality mapping is an involution, $(P^{\dl})^{\dl}=P$.
Setting $x^{\dl}=x$ for each variable $x$, the duality principle
extends to equations and it is easy to see that \CP\ is a self-dual axiomatisation:
\eqref{CP1} and \eqref{CP2} are each other's dual, and \eqref{CP3} and~\eqref{CP4}
are self-dual (i.e., identical to their own duals). Hence, 
\[\text{For all terms $s,t$ over $\SigCP,~\CP\vdash s=t ~\iff~\CP\vdash s^\dl=t^\dl$}.\]

We define the subset of `basic forms' of \PS\  that can be used to prove that \CP\ is 
a complete set of axioms.

\begin{definition}
\label{def:basic}
\textbf{Basic forms over $A$} are defined by the following grammar
\[t::= \tr\mid\fa\mid t\lef a \rig t\quad\text{for $a\in A$.}\]
We write $\BF$ for the set of basic forms over $A$. 
\end{definition}

Basic forms\footnote{%
  We speak of `basic form' rather than `normal form' because the basic form associated with atom $a$
  is $\tr\lef a\rig \fa$, while  for normal forms, one could have expected the reverse.} can be seen as \emph{evaluation trees}, which are binary, rooted trees with internal nodes
labelled from $A$ and leaves in $\{\tr,\fa\}$:
a branch from the root to a leaf represents the process of evaluation and all internal nodes 
represent the evaluation of the atom it is labelled with, while the leaf represents the final evaluation result.
Also, $\tr$ and $\fa$ are seen as evaluation trees that represent the evaluation of the constants
$\tr$ and $\fa$, respectively. Typically, the evaluation tree associated with $\tr\lef a\rig\fa$ is
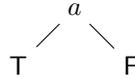
\[
\hspace{-12mm}
\begin{tikzpicture}[%
      level distance=7.5mm,
      level 1/.style={sibling distance=15mm},
      baseline=(current bounding box.center)]
      \node (a) {$a$}
        child {node (b1) {$\tr$}}
        child {node (b2) {$\fa$}}
        ;
      \end{tikzpicture}
\]

Evaluation trees were introduced in~\cite{Daan}, and a formal relation between their equality and that 
of the associated basic forms is established in~\cite{BP15}. However, basic forms themselves can be
represented as evaluation trees, as explained in the following example. 

\begin{example}
\label{ex1}
The basic form
$\fa\lef b\rig(\tr\lef a\rig\fa)$
can be represented as follows, where $\lef$ yields a left branch (the \true-case), 
and $\rig$ a right branch (the \false-case):
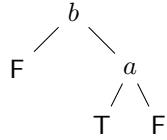
\[
\hspace{-12mm}
\begin{tikzpicture}[%
      level distance=7.5mm,
      level 1/.style={sibling distance=15mm},
      level 2/.style={sibling distance=7.5mm},
      baseline=(current bounding box.center)]
      \node (a) {$b$}
        child {node (b1) {$\fa$}
        }
        child {node (b2) {$a$}
          child {node (d1) {$\tr$}} 
          child {node (d2) {$\fa$}}
        };
      \end{tikzpicture}
\]
and expresses that if $b$ evaluates to \true\ (left branch), the expression evaluates to \false,
while if $b$ evaluates to \false\ (right branch), the evaluation of $a$ determines 
the overall evaluation result.
\end{example}

We now recall some definitions and results from~\cite{BP15}. In order to support
the intuitions, we spell out the proofs of Lemma~\ref{la:nieuw} and Theorem~\ref{thm:1a}
(in~\cite[La.2.17 and Thm.2.18]{BP15}).

\begin{lemma}
\label{la:2.5}
For each $P\in\PS$ there exists $Q\in\BF$ such that $\CP\vdash P=Q$.
\end{lemma}

\begin{definition}
\label{def:bf}
Given $Q,R\in\BF$, the auxiliary function $[\tr\mapsto Q, \fa\mapsto R]:\BF\to\BF$ 
for which postfix 
notation $P[\tr\mapsto Q, \fa\mapsto R]$ is used,
is defined as follows:
\begin{align*}
\tr[\tr\mapsto Q, \fa\mapsto R]&=Q,\\
\fa[\tr\mapsto Q, \fa\mapsto R]&=R,\\
(P_1\lef a\rig P_2)[\tr\mapsto Q, \fa\mapsto R]
&=P_1[\tr\mapsto Q, \fa\mapsto R]\lef a\rig P_2[\tr\mapsto Q, \fa\mapsto R].
\end{align*}
The \textbf{basic form function} $\baf:\PS \to\BF$ 
is defined as follows:
\begin{align*}
\baf(\tr)&= \tr,\\
\baf(\fa)&= \fa,\\
\baf(a)&=\tr\lef a\rig \fa\qquad\text{for all $a\in A$},\\
\baf(P \lef Q\rig R)&= \baf(Q)[\tr\mapsto \baf(P), \fa\mapsto \baf(R)].
\end{align*}
\end{definition}

The following lemma implies that $\baf()$ is a normalisation function; both statements 
easily follow by structural induction.

\begin{lemma}
\label{la:bf}
For all $P\in\PS$, $\baf(P)$ is a basic form, and
for each basic form $P$, $\baf(P)=P$.
\end{lemma}

\begin{definition}
\label{def:freevca}
The binary relation $=_\baf$ on \PS\
is defined as follows: 
\[P=_\baf Q~\iff~\baf(P)=\baf(Q).\]
\end{definition}

\begin{lemma}
\label{la:rephrase}{ }
The relation $=_\baf$ is a congruence relation.
\end{lemma}

Before proving that \CP\ is an axiomatization of the relation $=_\baf$,
we show that each closed instance of axiom~\eqref{CP4} satisfies $=_\baf$.

\begin{lemma}
\label{la:nieuw}
For all $P,P_1,P_2, Q_1, Q_2\in\PS$,
\begin{align*}
\baf(Q_1\lef(P_1\lef P\rig P_2)\rig Q_2)&
=\baf((Q_1\lef P_1\rig Q_2)\lef P\rig(Q_1\lef P_2\rig Q_2)).
\end{align*}
\end{lemma}

\begin{proof}
By definition, the lemma's statement is equivalent with 
\begin{align}
\nonumber
\big(\baf(P)&[\tr\mapsto \baf(P_1),\fa\mapsto\baf(P_2)]\big)\;[\tr\mapsto\baf(Q_1),\fa\mapsto\baf(Q_2)]\\
\nonumber 
&=\baf(P)[\tr\mapsto\baf(Q_1\lef P_1\rig Q_2),\fa\mapsto\baf(Q_1\lef P_2\rig Q_2)].
\end{align}
We prove this 
by structural induction on the form that $\baf(P)$ can have.
If $\baf(P)=\tr$, then 
\begin{align*}
\big(\tr&[\tr\mapsto \baf(P_1),\fa\mapsto\baf(P_2)]\big)\;[\tr\mapsto\baf(Q_1),\fa\mapsto\baf(Q_2)]\\
&=\baf(P_1)[\tr\mapsto\baf(Q_1),\fa\mapsto\baf(Q_2)] \\
&=\baf(Q_1\lef P_1\rig Q_2)\\
&=\tr[\tr\mapsto\baf(Q_1\lef P_1\rig Q_2),\fa\mapsto\baf(Q_1\lef P_2\rig Q_2)].
\end{align*}
The case $\baf(P)=\fa$ follows in a similar way. 

The inductive case $\baf(P)= R_1\lef a \rig R_2$ is trivial (by Definition~\ref{def:bf}). 
\end{proof}

\begin{theorem}
\label{thm:1a}
For all $P,Q\in\PS$, 
\(\CP\vdash P=Q~\iff~ P=_\baf Q.\)
\end{theorem}

\begin{proof}
$(\Rightarrow)$ By Lemma~\ref{la:rephrase}, 
$=_\baf$ is a congruence relation and it easily follows
that closed instances of the \CP-axioms $\eqref{CP1}-\eqref{CP3}$
satisfy $=_\baf$. By Lemma~\ref{la:nieuw}, closed instances of 
axiom~\eqref{CP4} also satisfy $=_\baf$.

\noindent$(\Leftarrow$) Assume $P=_\baf Q$. According to Lemma~\ref{la:2.5},
there exist basic forms
$P'$ and $Q'$ such that $\CP\vdash P=P'$ and $\CP\vdash Q=Q'$, so by $(\Rightarrow)$, 
$P'=_\baf Q'$ and thus $P'=Q'$. Hence, $\CP\vdash P=P'= Q'=Q$.
\end{proof}

The relation $=_\baf$ coincides with \emph{free valuation congruence}, which is
in~\cite{BP11} defined in terms of valuation algebras, and
in~\cite{BP15} in terms of evaluation trees. Basic forms 
have a 1-1-relation with evaluation trees, as shown by their pictorial representation  
in Example~\ref{ex1}. Evaluation trees were introduced
in~\cite{Daan} by a function $\mathrm{CE}()$ that assigns these trees to 
closed terms: the function $\mathrm{CE}()$ is very
comparable with the basic form function $\baf()$. 

We now present definitions of the left-sequential variants of the common propositional connectives.
The connective ${\leftand}$ is called \emph{left-sequential conjunction}, 
and the little circle in its symbol prescribes that the left-argument is evaluated 
first and that evaluation stops if it yields \false. This evaluation strategy
is called \emph{short-circuit evaluation}: 
evaluation stops as soon as the evaluation result is known.

We define the signature 
\(
\SigSCL=\{\leftand,\leftor,\neg,\tr,\fa,a\mid a\in A\}
\) 
where SCL stands for short-circuit logic.
Negation and sequential conjunction are defined in terms of the conditional connective: 
\begin{align}
\label{negdef}
\neg x &=\fa\lef x\rig \tr,\\
\label{anddef}
x\leftand y &=y\lef x\rig \fa.
\end{align}
Left-sequential disjunction
${\leftor}$ is defined by the following axiom:
\begin{equation}
\label{ordef}
x\leftor y=\neg(\neg x\leftand\neg y).
\end{equation}
Note that $\leftand$ and $\leftor$ are each others duals, thus 
$(P\leftand Q)^\dl=P^\dl\leftor Q^\dl$, and 
$(\neg P)^\dl=\neg (P^\dl)$.

Next, we extend $\CP$ with the
equations~\eqref{negdef}, \eqref{anddef}, and \eqref{ordef}, notation 
\[\CP(\neg,\leftand,\leftor).\] 
The following equations are easily proved in 
$\CP(\neg,\leftand,\leftor)$:
\begin{align}
\label{eq:eq0}
\neg\tr &=\fa,\\
\label{eq:eq2}
x\leftor y&=\tr\lef x\rig y,\\
\label{eq:eq1}
\neg\neg x&=x,
\end{align}
for example,
\begin{align*}
x\leftor y
&=\fa\lef((\fa\lef y \rig\tr)\lef (\fa\lef x\rig\tr)\rig\fa)\rig\tr
&&\text{by $\eqref{negdef} - \eqref{ordef}$}\\
&=\fa\lef(\fa\lef x\rig(\fa\lef y \rig\tr)) \rig\tr
&&\text{by \eqref{CP4}}\\
&=(\fa\lef\fa\rig\tr)\lef x\rig(\fa\lef(\fa\lef y \rig\tr) \rig\tr)
&&\text{by \eqref{CP4}}\\
&=\tr\lef x\rig y.
&&\text{by $\eqref{CP1} - \eqref{CP4}$}
\end{align*}

For the signature $\SigMix=\SigCP\cup\SigSCL$, let $\Terms$ be the set of all its closed terms. 
\begin{definition}
\label{def:mix}
The domain of the function $\baf()$ 
(Def.\ref{def:bf}) 
is extended
to \Terms\ as follows:
\\[2mm]\indent
\(
\baf(\neg P)=\baf(\fa\lef P\rig\tr),
~~\baf(P\leftand Q)=\baf(Q\lef P\rig\fa), 
~~\baf(P\leftor Q)=\baf(\tr\lef P\rig Q).
\)
\\[2mm]
The relation $=_\baf$ (Def.\ref{def:freevca}) is extended to \Terms.
\end{definition}
It follows 
easily that $=_\baf$ is a congruence and that
Theorem~\ref{thm:1a} can be generalised.
\begin{theorem}
\label{thm:1ab}
For all $P,Q\in\Terms$, 
\(\CP(\neg,\leftand,\leftor)\vdash P=Q~\iff~ P=_\baf Q.\)
\end{theorem}

\section{Memorising valuation congruence and mem-basic forms}
\label{sec:MVC}
In this section we discuss memorising valuation congruence, a congruence
obtained by extending \CP\ with one axiom. Then we recall a normalisation function
for so-called mem-basic forms.

\smallskip

We extend \CP\ as defined in Table~\ref{tab:CP} with the axiom
\begin{align}
\label{CPmem} 
\tag{CPmem}  
x \lef y \rig (z \lef u \rig (v \lef y \rig w)) &= x \lef y \rig (z \lef u \rig w).
\end{align} 
This axiom expresses that 
the first evaluation value of $y$ is \emph{memorised}. 
We write 
\[\CPmem\]
for this extension of \CP.
The dual of the axiom \eqref{CPmem} is easily derived in \CPmem:
\begin{align*}
((w\lef y\rig v)\lef u\rig z)\lef y\rig x&=x\lef (\fa\lef y\rig\tr)\rig(z\lef 
(\fa\lef u\rig\tr)\rig(v\lef (\fa\lef y\rig\tr)\rig w))\\
&=x\lef (\fa\lef y\rig\tr)\rig(z\lef (\fa\lef u\rig\tr)\rig w)
\hspace{1.2cm}
\text{by~\eqref{CPmem}}\\
&=(w\lef u\rig z)\lef y\rig x,
\end{align*}
so \CPmem\ also satisfies the duality principle: 
\[\text{For all terms $s,t$ over $\SigCP$,
$\CPmem\vdash s=t~\iff~\CPmem\vdash s^{\dl}=t^{\dl}$.}\]
We note that in \CPmem\ other variants of the axiom~\eqref{CPmem} are derivable, such as  
\begin{align}
\label{mem1}
(x \lef y\rig(z\lef u\rig v))\lef u\rig w = (x \lef y \rig z)\lef u \rig w, 
\end{align}
and that \emph{contraction}
is also derivable: 
\begin{align}
\label{con1}
(x\lef y\rig z)\lef y\rig u&=x\lef y\rig u,
\\
\label{con2}
x\lef y\rig (z\lef y\rig u)&=x\lef y\rig u.
\end{align}

We define a proper subset of basic forms with the property that each
propositional statement can be proved equal to such a basic form. 

\begin{definition}
\label{def:membf}
\textbf{Mem-basic
forms over $A$} are inductively defined:
\begin{itemize}
\item \tr\ and \fa\ are mem-basic forms, and
\item For $a\in A$, $P\lef a\rig Q$ is a mem-basic form if $P$
 and $Q$ are mem-basic forms in which $a$ does not occur.
\end{itemize}
We write 
\(\MBF\)
for the set of mem-basic forms over $A$.
\end{definition}

Note that if $A$ is finite, the number of mem-basic forms is also finite.
Mem-basic forms characterise evaluation trees with the property that in 
each path from the root to a leaf, the internal nodes have distinct labels.

The following normalisation function transforms closed terms to basic forms and then strips 
off repeated occurrences of atoms by auxiliary functions
$\ell_a()$ and $\ri_a()$ that can be called `left-$a$-reduction' and `right-$a$-reduction', respectively.

\begin{definition}
\label{def:membfu}
The \textbf{mem-basic form function}
$\membf:\PS\to \MBF$
is defined by
\begin{align*}
\membf(P)&=\memf(\baf(P)).
\end{align*}
The auxiliary function $\memf:\BF\to\BF$ is defined inductively:
\begin{align*}
\memf(\tr)&=\tr,~
\\
\memf(\fa)&=\fa,~
\\
\memf(P\lef a\rig Q)&=\memf(\ell_a(P))\lef a\rig\memf(\ri_a(Q)).
\end{align*}
For $a\in A$, the auxiliary functions
$\ell_a: \BF\to\BF$ and $\ri_a: \BF\to\BF$ 
are defined by
\begin{align*}
&\ell_a(B)=\ri_a(B)=B\text{ if $B\in\{\tr,\fa\}$, and~}
\begin{cases}
\ell_a(P\lef b\rig Q)=
\begin{cases}
\ell_a(P)&\text{if } b=a,\\
\ell_a(P)\lef b\rig \ell_a(Q)&\text{otherwise},
\end{cases}
\\[8mm]
\ri_a(P\lef b\rig Q)=\begin{cases}
\ri_a(Q)&\text{if } b=a,\\
\ri_a(P)\lef b\rig \ri_a(Q)&\text{otherwise}.
\end{cases}
\end{cases}
\end{align*}
\end{definition}

It is not hard to see that $\membf(P)\in\MBF$ for each $P\in\PS$.
As an example we depict the basic form $((\fa\lef a\rig\tr)\lef b\rig\fa)\lef a\rig\fa$
and its \membf-image $(\fa\lef b\rig\fa)\lef a\rig\fa$:
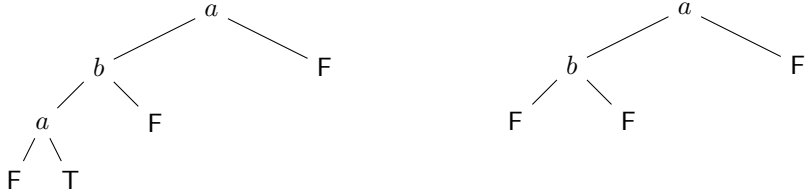
\[
\begin{array}{ll}
\begin{array}{l}
\begin{tikzpicture}[%
level distance=7.5mm,
level 1/.style={sibling distance=30mm},
level 2/.style={sibling distance=15mm},
level 3/.style={sibling distance=7.5mm}
]
\node (a) {$a$}
  child {node (b1) {$b$}
    child {node (c1) {$a$}
      child {node (d1) {$\fa$}} 
      child {node (d2) {$\tr$}}
    }
    child {node (c2) {$\fa$}
    }
  }
  child {node (b2) {$\fa$}
  };
\end{tikzpicture}
\end{array}
&\qquad
\begin{array}{l}
\qquad
\begin{tikzpicture}[%
level distance=7.5mm,
level 1/.style={sibling distance=30mm},
level 2/.style={sibling distance=15mm},
level 3/.style={sibling distance=7.5mm}
]
\node (a) {$a$}
  child {node (b1) {$b$}
    child {node (c1) {$\fa$}
    }
    child {node (c2) {$\fa$}
    }
  }
  child {node (b2) {$\fa$}
  };
\end{tikzpicture}
\\[8mm]
\end{array}\end{array}
\]

\begin{definition}
\label{def:memvc}
The binary relation $=_\membf$ on \PS\
is defined as follows: 
\[P=_\membf Q~\iff~\membf(P)=\membf(Q).\]
\end{definition}

In~\cite[Thm.5.9]{BP15} it is proved that \membf()\ is a normalisation function with the following property:
\begin{theorem}
\label{thm:membf}
For all $P,Q\in\PS$, $\CPmem\vdash P=Q~\iff~ P=_\membf Q$.
\end{theorem}

From a more general point of view, a mem-basic form represents a  
\emph{decision tree}, that is
a labelled, rooted, binary tree with internal nodes labelled from 
$A$ and leaves labelled from $\{\tr,\fa\}$ such that
for any path from the root to a leaf, the internal nodes receive distinct 
labels~\cite{Moret}.

The relation $=_\membf$ coincides with \emph{memorising valuation congruence}, 
which is in~\cite{BP11} defined in terms of so-called memorising valuation algebras, 
and in~\cite{BP15} in terms of so-called memorising evaluation trees. 
The rightmost tree displayed above exactly represents the memorising evaluation
tree of the conditional statement 
\[((\fa\lef a\rig\tr)\lef b\rig\fa)\lef a\rig\fa,\]
and of course also the memorising evaluation
tree of $(\fa\lef b\rig\fa)\lef a\rig\fa$.

Using Definition~\ref{def:mix}, 
we extend both the domain of the function $\membf()$ 
(Def.\ref{def:membfu}) 
and the relation $=_\membf$ to \Terms.
It follows that $=_\membf$ is a congruence and that Theorem~\ref{thm:membf} can be generalised.
\begin{theorem}
\label{thm:membfb}
For all $P,Q\in\Terms$, 
\(\CPmem(\neg,\leftand,\leftor)\vdash P=Q~\iff~ P=_\membf Q.\)
\end{theorem}

\section{\MSCL, Memorising Short-Circuit Logic}
\label{sec:mscl}
In this section we consider the set \MSCLe\ of equational axioms in 
Table~\ref{tab:MSCL} 
and recall the fact that \MSCLe\ axiomatises \MSCL, that is, memorising short-circuit
logic.

\smallskip

\begin{table}
{
\centering
\rule{1\textwidth}{.4pt}
\begin{align}
\label{Neg}
\tag{Neg}
\fa&=\neg\tr
\\
\label{Or}
\tag{Or}
x\leftor y&=\neg(\neg x\leftand\neg y)
\\ 
\label{Tand}
\tag{Tand}
\tr\leftand x&=x
\\ 
\label{Abs}
\tag{Abs}
x\leftand(x\leftor y)&=x\\
\label{Mem}
\tag{Mem}
(x\leftor y)\leftand z
&=(\neg x\leftand (y\leftand z))\leftor (x\leftand z)
\end{align}
\hrule
}
\caption{$\MSCLe$, a complete, independent set of axioms for \MSCL}
\label{tab:MSCL}
\end{table}

Axioms~\eqref{Neg} and~\eqref{Or} are explained in Section~\ref{sec:2},
and axiom~\eqref{Tand} needs no explanation.
Axiom~\eqref{Abs} is a left-sequential
variant of the absorption law and captures a first aspect of 
memorising valuation congruence: if
$x$ evaluates to \false, then this axiom holds, and if $x$ evaluates to \true, then
its second evaluation does so as well, and prevents evaluation of $y$. 

Axiom~\eqref{Mem} captures another, less obvious aspect of 
memorising valuation congruence: if $x$ evaluates to \true, then $z$ 
determines the evaluation result of both expressions because the evaluation 
result of $x$ is memorised;
if $x$ evaluates to \false, the evaluation result of both expressions is 
determined by $y\leftand z$ because the right disjunct $(x\leftand z)$
also evaluates to \false\ (because the evaluation 
result of $x$ is memorised). 

We define $\CPmem(\neg,\leftand,\leftor)$ as $\CPmem$ extended with the
equations~\eqref{negdef}, \eqref{anddef}, and \eqref{ordef}. 
We note that in $\CPmem(\neg,\leftand,\leftor)$, Hoare's conditional 
connective can be defined:
\begin{align}
\nonumber
(x\leftand y)\leftor (\neg x \leftand z)
&= \tr\lef(y\lef x\rig\fa)\rig(z\lef(\fa\lef x\rig\tr)\rig \fa)
&&\text{by~\eqref{negdef}-\eqref{ordef}}\\
\nonumber
&= \tr\lef(y\lef x\rig\fa)\rig(\fa\lef x\rig z)
&&\text{by~\eqref{CP4}, \eqref{CP2}, \eqref{CP1}}\\
\nonumber
&= (\tr\lef y\rig(\fa\lef x\rig z))\lef x\rig(\fa\lef x\rig z)
&&\text{by~\eqref{CP4}, \eqref{CP2}}\\
\label{eq:defcorres}
&=y\lef x\rig z.
&&\text{by~\eqref{mem1}, \eqref{CP3}, \eqref{con2}}
\end{align}
Furthermore, as a simple example of contraction (equation~\eqref{con1}), axiom \eqref{Abs}
easily follows from $\CPmem(\neg,\leftand,\leftor)$: 
$x\leftand(x\leftor y)=(\tr\lef x\rig y)\lef x\rig\fa=\tr\lef x\rig\fa=x$.

In~\cite{BPS13}, \emph{memorising short-circuit logic}, 
notation \MSCL, is defined as the equational logic that implies 
that part of the equational theory 
of $\CPmem(\neg,\leftand,\leftor)$ that is expressed in $\SigSCL$.
In~\cite{BPS21} it is proved that 
\MSCL\ is axiomatised 
by \MSCLe\ as defined in Table~\ref{tab:MSCL}, thus
\begin{equation}
\label{corres1}
\text{For all (open) terms $s,t$ over $\SigSCL, ~\MSCLe\vdash s=t\iff\MSCL\vdash s=t$.}
\end{equation}
From the correspondence result~\eqref{corres1} and Theorem~\ref{thm:membfb} it follows 
that ${\leftand}$ is not commutative: 
\[\membf(a\leftand \fa)=\membf(\fa\lef a\rig\fa)=\fa\lef a\rig\fa\neq\fa=\membf(\fa\leftand a),\]
and that for closed terms over $\SigSCL$, \MSCLe\ axiomatises $=_\membf$.

\begin{table}
{
\centering
\rule{1\textwidth}{.4pt}
\begin{align}
\tag{F3}
\label{F3}
\neg \neg x &= x
\\
\tag{F5}
\label{F5}
x \leftand \tr &= x
\\
\tag{F6}
\fa \leftand x &= \fa
\\
\tag{F7}
\label{ax5}
(x \leftand y) \leftand z &= x \leftand (y \leftand z) 
\\
\tag{F8}
\neg x \leftand \fa &= x\leftand\fa
\\
\tag{F9}
\label{F9}
(x \leftand \fa)\leftor y &= (x\leftor\tr)\leftand y
\\
\label{F10}
\tag{F10}
(x \leftand y) \leftor (z \leftand \fa) &= (x \leftor (z \leftand \fa)) 
\leftand (y \leftor (z \leftand \fa))
\\[2mm]
\label{C1}
\tag{C1}
x \leftand (y \leftand x) &= x \leftand y
\\
\label{C2}
\tag{C2}
x \leftand (y \leftand \neg x) &= x \leftand (y \leftand \fa)
\\[2mm]
\label{M1}
\tag{M1}
(x \leftand y) \leftor (\neg x \leftand z) 
&= (\neg x \leftor y) \leftand (x \leftor z)
\\
\label{M2}
\tag{M2}
(x \leftand y) \leftor (\neg x \leftand z) 
&= (\neg x \leftand z) \leftor (x \leftand y)
\\
\label{M3}
\tag{M3}
((x\leftand y)\leftor(\neg x\leftand z))\leftand u
&=(x\leftand (y\leftand u))\leftor(\neg x\leftand (z\leftand u))
\\[2mm]
\label{Dis}
\tag{Dis}
x \leftand (y \leftor z) &= (x \leftand y) \leftor (x \leftand z)
\end{align}
\hrule
}
\caption{Some consequences of $\MSCLe$, where~\eqref{Dis} stands for left-distributivity }
\label{tab:consequences}
\end{table}

Some nice and natural consequences of \MSCLe\ are collected in Table~\ref{tab:consequences} (more
of them in~\cite{BPS21}).
Note that with~\eqref{F3}, \eqref{M1} and \eqref{M2} it follows that
\begin{equation}
\label{eq:corres}
(x \leftand y) \leftor (\neg x \leftand z) = (x \leftor z) \leftand (\neg x \leftor y),
\end{equation}
which, in addition to equation~\eqref{eq:defcorres}, 
provides another, simple definition of the conditional connective in terms of the Boolean 
connectives.
Next, with~\eqref{eq:corres}, \eqref{M1} and $\eqref{M3}^\dl$ it follows that
\[((x\leftand y)\leftor(\neg x\leftand z))\leftor u
=(x\leftand (y\leftor u))\leftor(\neg x\leftand (z\leftor u)).\]

As was proved in~\cite{PS18}, the axioms~\eqref{Neg}, \eqref{Or} and \eqref{Tand} from Table~\ref{tab:MSCL} together 
with~\eqref{F3} and $\eqref{F5} - \eqref{F10}$ from Table~\ref{tab:consequences} axiomatise \FSCL\ for closed terms
(where \FSCL\ is the abbreviation of \emph{Free short-circuit logic}, see further Section~\ref{sec:conc}
-- Related work).
In~\cite{PS18}, the names (F1), (F2) and (F4) are used for the axioms~\eqref{Neg}, \eqref{Or} and \eqref{Tand}, respectively.

\section{The left-sequential biconditional connective}
\label{sec:bico}
As is well-known, the biconditional is the logical connective that requires both of its arguments to evaluate
to the same truth value to return true, so that
\[x\leftrightarrow y=(x\wedge y)\vee (\neg x\wedge\neg y).
\] 
In this section we discuss a left-sequential variant of the biconditional connective $\leftrightarrow$.
We write $\ell$IFF (``left-iff") for this connective and use the notation
\[\liff\]
to mark that short-circuit evaluation is prescribed.\footnote{%
  As mentioned before, the  circle in the $\liff$ symbol indicates that the left argument must 
  be evaluated first and prescribes short-circuit evaluation.
  Observe that to determine the value of the expression $x\liff y$,
  the second argument $y$ must always be evaluated.
  For left-sequential conjunction this is different, the variant of left-sequential conjunction 
  that prescribes so-called \emph{full} sequential evaluation (as opposed to short-circuit evaluation) has
  notation ${\fulland}$ and always evaluates both conjuncts from left to right (indicated by the black circle),
  and can be defined by $x \fulland y = (x \leftor (y \leftand \fa)) \leftand y$, or, by
  $x \fulland y = y \lef x \rig (\fa \lef y \rig \fa)$ (see~\cite{Daan}).
} 
We provide axioms for this extension, prove a correspondence result, and briefly discuss duality.

\smallskip

In the setting of Hoare's ternary conditional connective, $\ell$IFF is easy to define.

\begin{definition}
\label{def:liff} The connective $\liff$ is in \CP\ defined by the axiom
$x\liff y=y\lef x\rig (\fa\lef y\rig\tr)$.
\end{definition}
So, $x\liff y$ only evaluates to \true\ if $x$ and $y$ evaluate to the same truth value.
In $\CPmem(\neg,\leftand,\leftor)$ it follows with equation~\eqref{eq:defcorres}, i.e.,
$x\lef y\rig z = (y\leftand x)\leftor(\neg y\leftand z)$, 
that
\[x\liff y=(x\leftand y)\leftor(\neg x\leftand\neg y).
\]

\begin{table}
\hrule
\begin{align*}
x\leftor y&=\neg(\neg x\leftand\neg y)
\tag{\ref{Or}}
\\[0mm]
x\leftand (x\leftor y)&=x
\tag{\ref{Abs}}
\\[0mm]
\label{Assoc}
\tag{Assoc}
(x\leftand y)\leftand z&=x\leftand(y\leftand z)
\\[2mm]
\label{Iff1}
\tag{Tx}
\tr\liff x&=x
\\[0mm]
\label{Iff2}
\tag{xF}
x\liff \fa&=\neg x
\\[0mm]
\label{Iff3}
\tag{AndIff}
(x\leftand y)\liff z&=(x\leftand(y\liff z))\leftor (\neg x\leftand \neg z)
\end{align*}
\hrule
\caption{\MSCLea, a complete, independent set of axioms for \MSCLea}
\label{tab:leftIFF}
\end{table}

Memorising valuation congruence is axiomatised by \MSCLea, the set of axioms in Table~\ref{tab:leftIFF}
over the signature $\SigIff=\SigSCL\cup\{\liff\}$. 
The axiomatisation result for \MSCL\ mentioned in equation~\eqref{corres1} gives rise to the following result.
\begin{theorem}
\label{thm:c1}
~\\[1mm]
$1.$ $\MSCLea \vdash \MSCLe$, 
\\[1mm]
$2.$ $\MSCLe\cup\{x\liff y=(x\leftand y)\leftor (\neg x\leftand \neg y)\}\vdash \MSCLea$, 
\\[1mm]
$3.$ The axioms of \MSCLea\ in Table~\ref{tab:leftIFF} are independent.
\end{theorem}

\begin{proof}~
\\[-8mm]
\begin{enumerate}
\item
Axiom~\eqref{Neg}, thus $\fa=\neg\tr$, follows immediately from 
axioms~\eqref{Iff1} and \eqref{Iff2}.

An auxiliary result is $\neg\neg x=x$ (DNS, the double negation shift),
which follows with Prover9.

Axiom~\eqref{Tand}, that is $\tr\leftand x=x$, follows with Prover9 (faster if DNS 
is added to \MSCLea).

Axiom~\eqref{Mem}, that is $(x\leftor y)\leftand z
=(\neg x\leftand (y\leftand z))\leftor (x\leftand z)$, follows with Prover9:
if DNS and~\eqref{Tand} are added to \MSCLea, a proof
with option kbo is relatively fast. 

\item With Prover9 (immediate).

\item Independence of \MSCLea\ follows immediately with Mace4. The only case to be mentioned is the 
independence of axiom~\eqref{Assoc}: a counter model (with domain size 8) is quickly 
generated by Mace4 if (DNS) is added to $\MSCLea\setminus\{\eqref{Assoc}\}$. 
\end{enumerate}
\vspace{-3mm}
\end{proof}

Some nice consequences of \MSCLea, all checked with Prover9: \footnote{%
  For the last consequence, a run in Prover9 with the extra assumption 
  $\neg\neg x=x$ and options kbo and fold is relatively fast.}
\begin{align*}
\tr \liff x &= x\liff \tr\\
\fa\liff x &= x\liff \fa\\
\neg (x\liff  y)&=x\liff \neg y \\
\neg x\liff \neg y&=x\liff y 
\\
(x\leftor\tr)\liff y&= (x\leftor\tr)\leftand y.\\
(x\liff y)\liff z&= x\liff (y\liff z)
\end{align*}
With reference to correspondence result~\eqref{corres1}, we note that these consequences can also be easily proved in $\CPmem(\neg,\leftand,\leftor)$
extended with the equation $x\liff y=y\lef x\rig \neg y$.
A consequence of $\MSCLea$ that is more difficult to prove with Prover9,\footnote{%
  With all other \MSCLe-axioms added to \MSCLea, a run in Prover9 with options kbo and fold 
  required 112 CPU seconds.}
is recorded in~\cite[dual of axiom (AX4)]{Cor20}:
\[
(x\liff y)\leftand (z\leftand \fa)=(x\leftor(\neg y\leftand (z\leftand \fa)))\leftand (y\leftand (z\leftand \fa)).
\]
However, there is a simple proof in $\CP(\neg,\leftand,\leftor)$
extended with $x\liff y=y\lef x\rig \neg y$. 

\paragraph{Duality and $\ell$IFF.}
The duality principle can be extended to terms over $\SigIff$ if the left-sequential 
version of the connective \emph{exclusive or} ({XOR}, notation $\oplus$) is added as the dual of $\ell$IFF.  
We write $\ell$XOR (``left-xor") for this connective and use the symbol $\lxor$.  

Define $\SigIffD=\SigIff\cup\{\lxor\}$. 
So, for all $P, Q\in\SigIffD$,
\[P\lxor Q = (P^\dl\liff Q^\dl)^\dl.\]
With the definining axiom for $\liff$, that is, $x\liff y=y\lef x\rig(\fa\lef y\rig\tr)$ we find 
\[P\lxor Q = (Q^\dl\lef P^\dl\rig(\fa\lef Q\rig\tr)^\dl)^\dl=
(\fa\lef Q\rig\tr)\lef P\rig Q,\]
and thus $P\lxor Q=P\liff\neg Q$.
With respect to the signature $\SigIffD$, the axiom
\[x\lxor y=x\liff \neg y\]
defines $\lxor$ and we write \MSCLead\ for the addition of this axiom to $\MSCLea$ (see Table~\ref{tab:leftIFF}).
Hence, in \MSCLead\ it follows that
\[\neg(x\lxor y)=\neg(x\liff \neg y)=x\liff y,\]
so $\liff$ is also ``the negation of $\lxor$".
From \MSCLead\ and the \MSCLe-identities discussed in Section~\ref{sec:mscl}
it easily follows that 
\[x\lxor y=(x\leftand \neg y)\leftor(\neg x \leftand y).\]

\begin{theorem}[Duality]
For all terms s,t over $\SigIffD$,
\[\MSCLead\vdash s=t \iff~\MSCLead\vdash s^\dl=t^\dl.\]
\end{theorem}

Of course, we could have started with adding $\ell$XOR to $\SigSCL$ instead of $\ell$IFF.
It appears that taking all duals of the axioms~$\eqref{Iff1} - \eqref{Iff3}$
in Table~\ref{tab:leftIFF}, that is,
\begin{align*}
\fa\lxor x&=x&&\eqref{Iff1}^\dl\\
x\lxor \tr&=\neg x&&\eqref{Iff2}^\dl\\
(x\leftor y)\lxor z&=(x\leftor(y\lxor z))\leftand (\neg x\leftor \neg z)&&\eqref{Iff3}^\dl
\end{align*}
yields an axiomatisation that is also independent and
has the same equational theory (modulo $x\lxor y=x\liff \neg y$).

\section{The left-sequential NAND connective}
\label{sec:nand}
The Sheffer stroke $\mid$, also known as the NAND (not and) connective, requires at least one of its
arguments to be false so that it returns true:
\[x \mid y = \neg(x \wedge y).\]
In this section we discuss a left-sequential variant of the NAND connective. 
We write $\ell$NAND (“left-NAND”) for this connective and use the notation
\[\lnand\]
to mark that a left-sequential evaluation strategy is prescribed. We provide axioms for this
extension and prove a correspondence result.

\smallskip

In~\cite{Cor20}, the connective $\ell\text{NAND}$ is defined in \CP\ by
\[
x \lnand y=(\fa\lef y\rig\tr)\lef x\rig\tr,
\]
and hence satisfies $x \lnand y=\neg (x\leftand y)=\fa\lef(y\lef x\rig\fa)\rig\tr=(\fa\lef y\rig\tr)\lef x\rig\tr$.
Conversely, negation and the sequential connectives can be defined in terms of $\ell$NAND by
\begin{align}
\label{1l}
\tag{1$\ell$}
\neg x &= x \lnand \tr,
\\
\label{2l}
\tag{2$\ell$}
x \leftand y &= (x \lnand y) \lnand \tr,
\\
\label{3l}
\tag{3$\ell$}
x \leftor y &= (x \lnand \tr) \lnand (y\lnand\tr).
\end{align}
To obtain the $\ell$NAND-translation of the EqMSCL-axioms, we apply 
these definitions to each occurrence of the 
$\{\neg, \leftand , \leftor\}$-connectives for every axiom of \MSCLe:
\begin{align}
\fa &= \tr \lnand \tr,
\\
\label{elf}
(x \lnand\tr) \lnand(y \lnand\tr)&=(((x \lnand\tr) \lnand(y \lnand\tr)) \lnand\tr) \lnand\tr,
\\
\label{twaalf}
(\tr\lnand x)\lnand\tr &= x,
\\
\label{dertien}
(x \lnand ((x \lnand\tr) \lnand(y \lnand\tr))) \lnand\tr&=x,
\\
\nonumber
(((x \lnand\tr) \lnand(y \lnand\tr)) \lnand z)\lnand\tr
&=\\
\label{veertien}
((((x \lnand\tr) &\lnand((y \lnand z) \lnand\tr)) \lnand \tr)\lnand\tr)\lnand(((x \lnand z)
\lnand\tr)\lnand\tr).
\end{align}
Next, we simplify the translation equations~$\eqref{elf}-\eqref{veertien}$ as follows:
\begin{align}
\label{vijftien}
(x\lnand\tr)\lnand\tr&=x
\quad\hspace{6.8mm}\text{\small[in~\eqref{elf}, replace $(x \lnand\tr) \lnand (y \lnand\tr)$ by $x$]},
\\[1mm]
\label{zestien}
\tr\lnand x &= x \lnand \tr
\quad
\text{\small[in~\eqref{twaalf}, add $(\_)\lnand\tr$ and apply~\eqref{vijftien}],}
\\[1mm]
\label{zeventien}
x \lnand((x \lnand\tr)\lnand y)&=x \lnand\tr
\\
\nonumber
& 
\text{\small[in~\eqref{dertien}, add $(\_)\lnand\tr$, replace $(y \lnand \tr)$ by $y$, and apply~\eqref{vijftien}]},
\\[1mm]
\label{achttien}
(x \lnand(y \lnand\tr)) \lnand z
&=((x \lnand((y \lnand z) \lnand\tr)) \lnand((x \lnand\tr) \lnand z))\lnand\tr
\\
\nonumber
\quad
\text{\small[in~\eqref{veertien}, add $(\_)\lnand\tr$, replace}
&\text{\small$(x \lnand \tr)$ by $x$ and the rightmost $x$ by $(x \lnand \tr)$, and apply~\eqref{vijftien}].}
\end{align}
Finally, if in equation~\eqref{zeventien} we replace the leftmost $x$ by $(\tr\lnand x)$ and the two occurrences 
$(x\lnand\tr)$ by $x$, then
equations~\eqref{vijftien} and \eqref{zestien} are consequences of the three remaining axioms, 
this follows easily with Prover9.
These three axioms are listed in Table~\ref{tab:nand} and we call this set of axioms $\NANDe$.

\begin{theorem}
\label{thm:c2}
~\\[1mm]
$1.$ $\NANDe\cup\{\eqref{1l}, \eqref{2l}, \eqref{3l}\}
\vdash \MSCLe$, 
\\[1mm]
$2.$ $\MSCLe\cup\{x\lnand y=\neg(x\leftand y)\}\vdash \NANDe$, 
\\[1mm]
$3.$ The axioms of \NANDe\ in Table~\ref{tab:nand} are independent.
\end{theorem}

\begin{proof}~
\\[-8mm]
\begin{enumerate}
\item
This follows easily with Prover9. Recall equation \eqref{ordef}: $x\leftor y=\neg(\neg x\leftand\neg y)$. 
We note that 
\[\NANDe\cup\{\eqref{1l}, \eqref{2l}, \eqref{ordef}\}
\vdash \MSCLe\]
also follows easily with Prover9. 
\item
This follows easily with Prover9.
\item
By Theorem~\ref{thm:c2u} (which states that a superset of \NANDe\ is independent).
\vspace{-3mm}
\end{enumerate}
\end{proof}

\begin{table}
\hrule
\begin{align*}
\label{N1}
\tag{N1}
\fa&=\tr\lnand\tr\\[0mm]
\label{N2}
\tag{N2}
(\tr\lnand x)\lnand(x\lnand y)&=x\\[0mm]
\label{N3}
\tag{N3}
(x\lnand (y \lnand \tr)) \lnand z &=  
((x \lnand ( (y \lnand   z) \lnand   \tr)) \lnand   (( x \lnand   \tr)\lnand z)) \lnand   \tr
\end{align*}
\hrule
\caption{$\NANDe$, a set of equational axioms for $\ell$NAND, left-sequential NAND}
\label{tab:nand}
\end{table}

A first consequence of \NANDe\ stems from the characterisation of $y\lef x\rig z$ in \MSCLea,
i.e., $y\lef x\rig z=(x\leftand y)\leftor(\neg x\leftand z)=(\neg x\leftand z)\leftor(x\leftand y)$
(equations~\eqref{eq:defcorres} and~\eqref{M1} from Table~\ref{tab:consequences}). With the $\ell$NAND
definitions of the Boolean connectives we find
\begin{align*}
(x\lnand y)\lnand ((x\lnand\tr)\lnand z)
&=((x\lnand\tr)\lnand z)\lnand(x\lnand y).
\end{align*}
Furthermore, with~\eqref{M1}, \eqref{M2} and $(x\lnand\tr)\lnand\tr=x$ 
it follows that
\[(x\lnand y)\lnand ((x\lnand\tr)\lnand z)=(((x\lnand\tr)\lnand(z\lnand\tr))\lnand(x\lnand(y\lnand\tr)))\lnand\tr.\]

\paragraph{Abbreviation.}
If we write $x'$ for $x\lnand\tr$, terms become more readable.
In the remainder of our discussion about $\lnand$ we will 
mostly use this abbreviation, e.g., $x''=x$. The characterisations of  
the conditional then look like this:
\begin{align}
\nonumber
y\lef x\rig z
&=(x\lnand y)\lnand (x'\lnand z)\\
\nonumber
&=(x'\lnand z)\lnand (x\lnand y)\\
\nonumber
&=((x'\lnand z')\lnand(x\lnand y'))'\\
\label{char}
&=((x\lnand y')\lnand(x'\lnand z'))'.
\end{align}

Some more consequences of \NANDe, all of which can be easily verified with Prover9: \footnote{%
  Verification with Prover9 is fastest if an auxiliary function $f(x)=x\lnand\tr$ is added 
  to the axioms of \NANDe, the options rpo and fold are used,
  and $f()$ and $()'$ are \emph{not} used in these consequences, i.e., the goal formulas use $x\lnand\tr$. 
  Without the addition of $f(x)=x\lnand\tr$ (or $x'=x\lnand\tr$), proofs of \eqref{b1}, \eqref{b4} and
  \eqref{b5} seem not feasible.
  }
\begin{align}
\label{b1}
((x\lnand y)\lnand(x'\lnand z))\lnand v&=(x\lnand (y\lnand v))\lnand(x'\lnand (z\lnand v)),\\
\label{b2}
\fa\lnand x&=\tr,\\
\label{b3}
x'\lnand(x'\lnand\fa)&=x,\\
\label{b4}
x\lnand((x\lnand y)\lnand(x'\lnand z))&=x\lnand y,\\
\label{b5}
x\lnand((y\lnand z)\lnand(y'\lnand w))&=x\lnand((y\lnand (x\lnand z'))\lnand(y'\lnand (x\lnand w'))). 
\end{align}

An advantage of the signature with $\ell$NAND is that
inductive properties are easier to prove than in the case of $\SigSCL$. 
We return to this point in the next section.

\paragraph{Duality and $\ell$NAND.} 
We write $\ell$NOR for the dual connective of $\ell$NAND. 
In ~\cite{Cor20},
the left-sequential version of the NOR connective, notation $\lnor$ (with notation $\downarrow$ for NOR),
is defined by
\[x\lnor y=\fa \lef x \rig (\fa \lef y \rig\tr).\]
Using \CP, it easily follows that $\lnor$ is expressible in terms of $\lnand$ , \tr\ and \fa:
\[x \lnor y=((x \lnand \tr) \lnand (y \lnand \tr)) \lnand\tr
\quad(= (x'\lnand y')', \text{ compare the symmetry with axiom\eqref{Or}}).\]
However, we see no reason why adding the dual connective $\lnor$ could be attractive, nor to prefer
$\lnor$ to $\lnand$.

\section{The three-valued case}
\label{sec:three}
In this section we discuss the addition of a third constant to \CPmem, and also to \MSCLea\ and \NANDe, 
which represents the truth value ``undefined". 
This is a follow-up to the results in~\cite{BPS21} on \MSCLu, that is, \MSCL\ with undefinedness.

The constant \und\ is used to represent the third truth value \emph{undefined}.
In the setting with the conditional connective, this constant is defined by the axiom
\begin{equation}
\label{eq:und}
\tag{CP-$\und$}
x\lef\und\rig y=\und, 
\end{equation}
which should be added to the axiom system under consideration. 
We write \CPmemu\ for the extension of \CPmem\ with axiom~\eqref{eq:und}.
With $\und^\dl= \und$, it follows that \CPmemu\ also satisfies the duality principle.
Let \PSu\ be the set of closed terms belonging to \CPmemu.

\begin{table}
\hrule
\begin{align*}
\label{Und}
\tag{Und}
&\MSCLeau:&\neg\und&=\und
\\[2mm]
\label{NU}
\tag{NU}
&\NANDeu:&
\und\lnand x&=\und
\end{align*}
\hrule
\caption{The axiom for \und\ in \MSCLeau\ and in \NANDeu}
\label{tab:others}
\end{table}

For each of the equational logics defined by \MSCLea\ and \NANDe, the additional axiom for \und\ is given 
in Table~\ref{tab:others}, and we write \MSCLeau\ and \NANDeu\ for
the resulting axiom systems.
It follows that 
\begin{align*}
&\MSCLeau\,\vdash\und\leftand x=\und\leftor x=\und\liff x=\und,\quad\fa\leftand\und=\fa,\\
&\NANDeu\vdash\fa\lnand\und=\tr.
\end{align*}
In~\cite{BPS21}, \emph{memorising short-circuit logic with undefinedness}, 
notation $\MSCLu$, is defined as the equational logic that implies the part of the equational theory 
of $\CPmemu(\neg,\leftand,\leftor)$ that is expressed in $\SigSCLu=\SigSCL\cup\{\und\}$,
and it is proved that \MSCLu\ is axiomatised 
by $\MSCLeu=\MSCLe\cup\{\eqref{Und}\}$ (see \cite[Thm.7.16]{BPS21}):
\[\text{For all terms $s,t$ over $\SigSCLu,~\MSCLeu\vdash s=t ~\iff~\MSCLu\vdash s=t$}.\]
It immediately follows that when \MSCLeu\ is replaced by \MSCLeau, this result is preserved. Moreover,
the axioms of \MSCLeau\ are independent: this follows quickly with Mace4.

In the remainder of this section, we will further discuss \NANDeu.

\begin{theorem}
\label{thm:c2u}
~\\[1mm]
$1.$ $\NANDeu\cup\{\eqref{1l}, \eqref{2l}, \eqref{3l}\}
\vdash \MSCLeu$, 
\\[1mm] 
$2.$ $\MSCLeu\cup\{x\lnand y=\neg(x\leftand y)\}\vdash \NANDeu$, 
\\[1mm]
$3.$ The axioms for \NANDeu\ in Tables~\ref{tab:nand} and~\ref{tab:others} are independent.
\end{theorem}

\begin{proof}
Statements 1 and 2 follow easily with Prover9, and 3 follows quickly with Mace4. 
\end{proof}

To illustrate an advantage of $\NANDeu$ over $\MSCLeu$ (or, $\MSCLeau$), we prove a representation 
result for basic forms. This still requires a lot of detail, but is simpler and more straightforward 
than the proof of the corresponding representation result for $\MSCLeu$ in~\cite[La.7.9 and La.7.10]{BPS21}).

\begin{definition}
\label{def:msnfu}
\textbf{Memorising \und-Nand Basic Forms $(\MSNFus)$ over $A$} are inductively defined:
\begin{itemize}
\item $\tr, \fa, \und$ are \MSNFus, 
and
\item For $a\in A$, $(a\lnand P)\lnand (a'\lnand Q)$ is a \MSNFu\  
if $P$ and $Q$ are \MSNFus\ that do not contain $a$.
\end{itemize}
We write \MFu\ for the set of all \MSNFus\ over $A$.
\end{definition}

The following functions on \MFu\ are used to compose \MSNFus.

\begin{definition}
\label{def:rmau}
For $a\in A$, the function $\tr_a^\und:\MFu\to\MFu$ is defined by 
\begin{align*}
\tr_a^\und(\tr)&=\tr,\quad
\tr_a^\und(\fa)=\fa,\quad
\tr_a^\und(\und)=\und,\\
\tr_a^\und((b\lnand P_1)\lnand (b'\lnand P_2))&=
\begin{cases}
P_1
&\text{if $b=a$},\\
(b\lnand \tr_a^\und(P_1))\lnand (b'\lnand \tr_a^\und(P_2))
&\text{otherwise}.
\end{cases}
\end{align*}
\\[1mm]
For $a\in A$ the function $\fa_a^\und:\MFu\to\MFu$ is defined by 
\begin{align*}
\fa_a^\und(\tr)&=\tr,\quad
\fa_a^\und(\fa)=\fa,\quad
\fa_a^\und(\und)=\und,\\
\fa_a^\und((b\lnand P_1)\lnand (b'\lnand P_2))&=
\begin{cases}
P_2
&\text{if $b=a$},\\
(b\lnand \fa_a^\und(P_1))\lnand (b'\lnand \fa_a^\und(P_2))
&\text{otherwise}.
\end{cases}
\end{align*}
\end{definition}

So, $\tr_a^\und$ removes the $a$-occurrences in a \MSNFu\ under the assumption that $a$ evaluates to \true,
and $\fa_a^\und$ does this under the assumption that $a$ evaluates to \false.
Note that for each $P\in\MFu$, both $\tr_a^\und(P)$ and $\fa_a^\und(P)$ are also \MSNFus.
To compose \MSNFus\ we use the following lemma.

\begin{lemma}
\label{la:composeu}
For all $a\in A$ and $P\in\MFu$, 
\\[1mm]
$1.$ $\NANDeu\vdash a\lnand P=a\lnand\tr_a^\und(P)$ and $a\lnand P'=a\lnand \tr_a^\und(P)'$,
\\[1mm]
$2.$ $\NANDeu\vdash a'\lnand P=a'\lnand\fa_a^\und(P)$
and $a'\lnand P'=a'\lnand \fa_a^\und(P)'$.
\end{lemma}

\begin{proof}
Statement 1 follows by induction on the structure of $P$. If $P \in\{\tr, \fa,\und\}$ this is trivial
because $\tr_a^\und(P)=P$.
For the induction step, four cases have to be dealt with: 
\\-
If $P = (a\lnand P_1) \lnand (a' \lnand P_2)$ 
then 
\[a\lnand P\stackrel{\eqref{b4}}= a\lnand P_1=a\lnand\tr_a^\und(P),\hspace{46mm}
\] 
and 
\begin{align*}
a\lnand P'
&=a\lnand((a\lnand P_1') \lnand (a' \lnand P_2'))
&&\text{by \eqref{char}}\\
&=a\lnand P_1'
&&\text{by~\eqref{b4}}\\
&=a\lnand\tr_a^\und(P_1)'
&&\text{by IH}\\
&=a\lnand\tr_a^\und(P)'.
\end{align*} 
\\-
If $P= (b\lnand P_1)\lnand(b'\lnand P_2)$ for $b\neq a$, then 
\begin{align*}
a\lnand P
&=a\lnand((b\lnand (a \lnand P_1'))\lnand(b'\lnand(a\lnand P_2')))
&&\text{by \eqref{b5}}\\
&=a\lnand((b\lnand (a \lnand \tr_a^\und(P_1)'))\lnand(b'\lnand(a\lnand \tr_a^\und(P_2)')))
&&\text{by IH}\\
&=a\lnand((b\lnand \tr_a^\und(P_1))\lnand(b'\lnand \tr_a^\und(P_2)))
&&\text{by \eqref{b5}}\\
&=a\lnand \tr_a^\und((b\lnand P_1)\lnand(b'\lnand P_2)),
\end{align*}
and 
\begin{align*}
a\lnand P'
&=a\lnand((b\lnand P_1')\lnand(b'\lnand P_2'))
&&\text{by \eqref{char}}\\
&=a\lnand((b\lnand (a \lnand P_1))\lnand(b'\lnand(a\lnand P_2)))
&&\text{by \eqref{b5}}\\
&=a\lnand((b\lnand (a \lnand \tr_a^\und(P_1)))\lnand(b'\lnand(a\lnand \tr_a^\und(P_2))))
&&\text{by IH}\\
&=a\lnand((b\lnand \tr_a^\und(P_1)')\lnand(b'\lnand \tr_a^\und(P_2)'))
&&\text{by \eqref{b5}}\\
&=a\lnand((b\lnand \tr_a^\und(P_1))\lnand(b'\lnand \tr_a^\und(P_2)))'
&&\text{by \eqref{char}}\\
&=a\lnand \tr_a^\und((b\lnand P_1)\lnand(b'\lnand P_2))'.
\end{align*}

Statement 2 follows in a similar way. 
\end{proof}

\begin{theorem}
\label{thm:msnfu}
For each term $P$ constructed from $\{\lnand, \tr, \fa,\und,a\mid a\in A\}$ there is $Q\in\MFu$ such that $\NANDeu\vdash P=Q$.
\end{theorem}

\begin{proof}
We first prove an auxiliary result: if $P$ and $Q$ are \MSNFus, there is a 
\MSNFu\ $R$ such that $\NANDeu\vdash P\lnand Q=R$.
We prove this by induction on the structure of $P$.
\begin{itemize}
\item
If $P=\tr$, make a case distinction on $Q$: if $Q\in\{\tr,\fa,\und\}$, this is trivial, and 
if $Q=(a \lnand Q_1)\lnand(a'\lnand Q_2)$,
then by induction there are \MSNFus\ $R_i$ such that $\NANDeu\vdash R_i=\tr\lnand Q_i$, and since 
$a$ does not occur in the $Q_i$, it does not occur in the $R_i$ either. Derive
\begin{align*}
\tr\lnand((a \lnand Q_1)\lnand(a'\lnand Q_2))
&=((a \lnand Q_1)\lnand(a'\lnand Q_2))'
&&\text{by~\eqref{zestien}}\\
&= (a \lnand Q_1')\lnand(a'\lnand Q_2')
&&\text{by~\eqref{b1}}\\
&= (a \lnand (\tr\lnand Q_1))\lnand(a'\lnand (\tr\lnand Q_2))
&&\text{by~\eqref{zestien}}\\
&= (a \lnand R_1)\lnand(a'\lnand R_2).
\end{align*}
\item If $P =\fa$, then $\fa\lnand Q=\tr$ by~\eqref{b2}.
\item If $P =\und$, then $\und\lnand Q=\und$ by axiom~\eqref{Und}.
\item If $P = (a \lnand P_1)\lnand(a'\lnand P_2)$, then by 
induction there are \MSNFus\ $R_i$ such that $\NANDeu\vdash R_i= P_i\lnand Q$. 
Note that $Q$, and therefore each of $R_i$, can contain the atom $a$.
Derive
\begin{align*}
((a \lnand P_1)\lnand(a'\lnand P_2))\lnand Q
&=(a\lnand (P_1\lnand Q))\lnand(a'\lnand (P_2\lnand Q))
&&\text{by~\eqref{b1}}\\
&=(a\lnand R_1)\lnand(a'\lnand R_2)\\
&=(a\lnand \tr_a^\und(R_1))\lnand(a'\lnand \fa_a^\und(R_2)).
&&\text{by Lemma~\ref{la:composeu}}
\end{align*}
\end{itemize}
This concludes the proof of the auxiliary result.

The theorem follows easily by structural induction on $P$. 
The base cases \tr, \fa, \und\ are trivial  
and by~\eqref{b3}, $\NANDeu\vdash a= (a\lnand\tr)\lnand (a'\lnand \fa)$. The inductive case follows 
from the auxiliary result.
\end{proof}

Finally, we note there is a one-to-one correspondence between \MBFu\ (the mem-basic forms 
that can also contain the constant \und) and \MFu: define $f$ and $g$ by
$f(n)=g(n)=n$ for $n\in\{\tr,\fa,\und\}$, 
$f(P\lef a\rig Q)=(a\lnand f(P))\lnand (a'\lnand f(Q))$, and
$g((a\lnand P)\lnand (a'\lnand Q))=g(P)\lef a\rig g(Q)$.
It easily follows that $g\circ f$ and $f\circ g$ are the identity on \MBF\ and \MFu, respectively.

\section{Conclusions}
\label{sec:conc}
We begin with a remark about the incorporation of the constants \tr\ and \fa\ 
in short-circuit logics, the logics that in programming model the use of conditions and 
prescribe short-circuit evaluation (more information and motivation can be found 
in~\cite{PS18,BPS21}). Non-commutative propositional logic, or MSCL, deals with the case 
where atomic side effects do not occur and requires incorporation of (at least one) of these
constants because they are not definable: 
for no atom $a$ does it hold that $a\leftor\neg a = \tr$ or $a\leftand\neg a = \fa$
(compare this with \CP\ and \CPmem, which require both constants). 
As for the ``non-commutativity" from the title of this paper, the difference between 
$a\leftand\fa$ and $\fa\leftand a$ is that the first expression requires evaluation of 
$a$ and the second does not, so these two expressions are not identified in MSCL.

This paper can be seen as a continuation of~\cite{BPS21} because it introduces
variants of \MSCL\ with additional or alternative connectives:
\begin{enumerate}
\item 
The addition of the connective $\ell$IFF to \MSCL\ (and/or its dual $\ell$XOR) can be 
motivated as a matter of convenience. The connectives $\ell$IFF and $\ell$XOR are definable 
in \MSCLe, so their addition is not essential and can only contribute to more comprehensible 
axioms and derivations. 

Conversely, negation is definable by the addition of any of these two by 
$\neg x =\fa \liff x$ (or $\neg x=x\liff\fa$) and $\neg x=\tr \lxor x$ (or $\neg x=x\lxor\tr$),
and thereby also the other connective is definable: $x \lxor y = x \liff \neg y=\neg(x\liff y)$. 
However, it does not seem an attractive idea to omit negation.

Finally, omitting $\leftand$ and $\leftor$ from \MSCLea, either one or 
both of $\liff$ and $\lxor$ is modulo memorising valuation congruence not 
sufficiently expressive, which can be easily seen by considering $a \leftand b$.

\item
A preference for \NANDe\ over \MSCLe\ can be motivated as a technical improvement.
As suggested in Section~\ref{sec:nand} (and illustrated in Section~\ref{sec:three}), 
the completeness result for \MSCLe\ proved in~\cite{BPS21} is easier to prove with help of \NANDe,
which has only three axioms and allows smaller and simpler inductive proofs
of the supporting lemmas. Then, Theorem~\ref{thm:c2} implies the completeness of \MSCLe. 
Similar remarks can be made for preferring \NANDeu\ to \MSCLeu.

\end{enumerate}

A second goal of this paper is to emphasise that the mem-basic
forms introduced in~\cite{BP11} themselves provide a semantics for closed terms:
this is based on the results in~\cite{BP15} on the correspondence between mem-basic
forms and evaluation trees, and on the coinciding congruences defined by each 
(see~\cite[Prop.5.13]{BP15}). 
It is not difficult to prove that this is also true for the extension to \Terms,
the set of closed terms belonging to $\CP(\neg,\leftand,\leftor)$ and $\CPmem(\neg,\leftand,\leftor)$,
compare Theorems~\ref{thm:1ab} and~\ref{thm:membfb}. 

We conclude with a reflection on the definition of \MSCL\ and the advantages and 
disadvantages of using Hoare's conditional as an auxiliary operator to define this 
short-circuit logic. 
For the question of which laws axiomatise short-circuit evaluation, the appeal to 
an auxiliary operator does not seem appropriate, and axioms for the Boolean short-circuit 
connectives are more interesting, so this can be considered a disadvantage.
On the other hand, the usefulness of the conditional connective as a means of proving 
properties of short-circuit logics\footnote{%
   Note that all short-circuit logics introduced in \cite{BPS13} are 
   defined with help of the conditional connective.} 
can be easily demonstrated.
A first advantage of $\CPmem(\neg,\leftand,\leftor)$ (the underlying set of axioms
of \MSCL)
over $\MSCLea$ and $\NANDe$ is that 
in its mem-basic forms (\MBF) each atom occurs at most once, which is also true of their \und-variants.
The analogue of Theorem~\ref{thm:msnfu} for \PSu-terms, viz,
\[
\text{For each $P\in\PSu$ there is $Q\in\MBFu$ such that $\CPmemu\vdash P=Q$,}
\]
is therefore much easier to prove (for the case without \und\ this is shown in~\cite[La.8.2]{BP11},
and \und\ does not provide a worrying extension here). 
As a second example, equation~\eqref{F9} in Table~\ref{tab:consequences}
has a simple proof in $\CPmem(\neg,\leftand,\leftor)$ (in which the axiom~\eqref{CPmem} is not needed):
\[(x\leftand \fa)\leftor y=\tr\lef(\fa\lef x\rig\fa)\rig y
=y\lef x\rig y=y\lef(\tr\lef x\rig\tr)\rig\fa=(x\leftor\tr)\leftand y,
\]
but a proof in \MSCLe\ is not so simple. 
As a last example, one can compare a proof of
the associativity of $\liff$ in $\CPmem\cup\{x\liff y=y\lef x\rig(\fa\lef y\rig\tr)\}$ 
with one in \MSCLea. Elaborating on this, there is a strong case for introducing 
the abbreviation $x'$ for $\fa\lef x\rig\tr$ in \CP, which yields equations like
\[(x\lef y\rig x')'=x\lef y'\rig x'\]
(which follow easily in \CP). This allows a very simple proof of the associativity of $\liff$ in 
$\CPmem\cup\{x\liff y=y\lef x\rig y'\}$ (again, the axiom~\eqref{CPmem} is not needed):
\begin{align*}
(x\liff y)\liff z=
z\lef(y\lef x\rig y')\rig z'
&=(z\lef y\rig z')\lef x\rig(z\lef y'\rig z')\\
&=(z\lef y\rig z')\lef x\rig(z\lef y\rig z')'
= x\liff(y\liff z).
\end{align*}

\paragraph{Related work.}

In 1948, Church introduced in~\cite{Chu48} the \emph{conditioned disjunction} $[p,q,r]$, 
which, following the author, may be read ``$p$ or $r$ according as $q$ or not $q$'' 
and which expresses exactly the same connective as Hoare's conditional (introduced in 1985). 
Church showed that this connective together with 
constants $\tr$ and $\fa$ form a complete set of independent primitive connectives for the propositional calculus. 
Church also noted that for propositional variables $a,b,c$, the dual of $[a,b,c]$ is simply $[c,b,a]$, so that
to dualize an expression of the propositional calculus in which the only connectives occurring are
conditioned disjunction, \tr, and \fa, it is sufficient to write the expression backwards and at the
same time to interchange \tr\ and \fa. 
For the conditioned disjunction, reference~\cite{Chu56} is often used, and also the name 
\emph{conditional disjunction}.
Although $[x,y,z]$ has explicit scoping, more complex expressions
such as $[[a,b,[b,a,c]],[a,b,[c,a,b],[a,b,c]]$
are difficult to read and for this reason we prefer Hoare's conditional connective $x\lef y\rig z$.

In 2013, we defined \emph{free short-circuit logic} (notation \FSCL) in~\cite{BPS13}
as the equational logic that implies the part of the equational theory 
of $\CP(\neg,\leftand,\leftor)$ that is expressed in $\SigSCL$.
In~\cite{PS18}, the set of axioms $\SCLe$ is defined as $\text{(F1)} - \eqref{F10}$, 
where $\text{(F1)}=\eqref{Neg}$, $\text{(F2)}=\eqref{Or}$, $\text{(F4)}=\eqref{Tand}$, 
and the remaining
axioms are listed in Table~\ref{tab:consequences}, and it is proved that $\SCLe$
axiomatises \FSCL\ for closed terms. We note that in \FSCL, not all basic forms can be expressed,
for example not those of $a\lef b\rig c$ and $a\lef a\rig (\fa\lef a\rig\tr)$
(see~\cite[Prop.12.1 and Thm.12.2]{BP11}).

In 2020, Cornets de Groot defined in~\cite{Cor20} the two left-sequential connectives 
$\ell$XOR and $\ell$NAND and studied their relation with \FSCL. 
In that paper, complete equational axiomatisations 
of the resulting logical systems for closed terms are provided, while also attention is paid to the dual
connectives $\ell$IFF and $\ell$NOR, respectively, and to expressiveness issues under free valuation
congruence.

\paragraph*{Acknowledgement.}
We thank Sven Cornets de Groot 
for pointing us to the conditioned disjunction defined in~\cite{Chu48,Chu56}.

\end{document}